\documentclass[11pt]{article}
\usepackage{eqnarray}
\usepackage{amsmath,amsthm}
\usepackage{amsfonts}
\usepackage{color}
\usepackage{marginnote}
\usepackage{graphics}
\usepackage{graphicx}
\usepackage[normalem]{ulem} 
\newtheorem{theorem}{Theorem}

\newtheorem{lemma}{Lemma}
\newtheorem{remark}{Remark}
\newtheorem{defin}{Definition}

\def \x {\mathbf{x}}
\def \dx {\mathbf{\dot{x}}}

\begin{document}
	\title{Stability analysis of a charged particle subject to two non-stationary currents}
	
	\author{Stefano Mar\`o\thanks{
			Department of Mathematics, Faculty of Science, University of Oviedo,
			Oviedo, Spain. E-mail: 
			\texttt{marostefano@uniovi.es}}
		\and
		Francisco Prieto-Castrillo\thanks{
			Department of Mathematics, Faculty of Science, University of Oviedo,
			Oviedo, Spain. E-mail: 
			\texttt{prietofrancisco@uniovi.es}
		}
	}
	\date{}
	\maketitle
	
	\begin{abstract}
		We study the non relativistic motions of a charged particle in the electromagnetic field generated by two parallel electrically neutral vertical wires carrying time depends currents. Under quantitative conditions on the currents we prove the existence of a vertical strip of stable motions of the particle. The stable strip is contained in the plane of the two wires and the stability is understood in a stronger sense than the isoenergetic stability of Hamiltonian systems. Actually, also variations of the integral given by the linear momentum will be allowed.
	\end{abstract}

\section{Introduction}
The study of charge dynamics in the presence of electromagnetic fields is a cornerstone of classical electrodynamics, with profound implications for both theoretical physics and practical engineering. This paper focuses on a specific scenario: the stability of a charge moving between two infinite conducting wires carrying non-uniform currents. The analysis of this problem not only enhances our understanding of  electromagnetic nonlinear interactions but also has significant applications in various advanced technological domains.
In particular, in magnetic confinement devices such as tokamaks and stellarators, controlling and stabilizing charged particles is crucial for maintaining hot and dense plasma. The study of charge stability under the Lorentz force can provide valuable insights into the design and operation of these devices, minimizing energy losses and maximizing confinement time \cite{wesson2004tokamaks}. This research could offer new perspectives on particle interactions with magnetic fields at the plasma edge, where instabilities are common and can lead to confinement loss \cite{connor2004edge}.\\
Also, turbulence is a major cause of anomalous transport in plasmas, affecting confinement efficiency. Understanding the stability of a charge in a complex magnetic field can help model how particles transport through the plasma. These findings could contribute to the development of more accurate theories and simulations of particle dynamics and the formation of coherent structures within the plasma \cite{diamond2005zonal}.\\
The principles of charged particle stability in magnetic fields also apply to the study of astrophysical plasmas, such as those found in the solar wind, planetary magnetospheres, and in the formation of stars and galaxies \cite{parker1958dynamics}.
To address these challenges, we employ the Lorentz force, which describes the force experienced by a charged particle in an electromagnetic field. We derive critical insights into the factors that influence the stability of a charge moving between two infinite conductors with non-uniform current distributions.

More precisely, we study the motion of a charged particle in the electromagnetic field generated by two parallel electrically neutral vertical wires carrying time depends currents. We deal with non relativistic scenario and employ the Newton-Lorentz equation (see \cite{Griffiths,Jackson,Lorentz}) given by:
\begin{equation}\label{eqLintro1}
	m\ddot{\mathbf{x}}= q(\mathbf{E}(t,\mathbf{x})+\dot{\mathbf{x}}\times \mathbf{B}(t,\mathbf{x})), \quad \mathbf{x}\in\mathbb{R}^3. 
\end{equation}
Here, $\mathbf{E}(t,\mathbf{x})$ and $\mathbf{B}(t,\mathbf{x})$ represent the electric and magnetic fields generated by the currents, being solutions of Maxwell's equations. 
In case of stationary currents, equations \eqref{eqLintro1} can be written as an autonomous Hamiltonian systems and has been extensively studied in \cite{ALP,GP1,GP2}. More precisely, the electric field vanishes and the kinetic energy is preserved. Together with the translational symmetry one can get interesting confinement properties. The simpler case of just one vertical wire enjoys also a rotational symmetry that makes the system integrable. The motion of particles with non-null angular momentum are helicoidal and can escape to infinity only in the vertical direction.        

The introduction of a time dependence in the currents brings the consequence that the energy is no longer preserved and the previous analysis cannot be carried on. Moreover, as far as the authors know,  the motion in time dependent electromagnetic field has not been studied widely. Some results have been obtained through variational techniques in relativistic regimes in \cite{ABT,ABT2}. However, in these works, electromagnetic fields are assumed to be regular. This is not satisfied by our problem since collisions with the wire produce singularities. In regard to isolated singularities, we cite \cite{GT} where periodic solutions are found. The case of an infinite wire has been studied in \cite{hau,king,lei_zhang}. However, unlike our case, a time-dependent charge density and no current are considered there.\\  
The problem of a single wire with time-dependent current has recently been studied in \cite{garzon}, where the existence of a stability zone is demonstrated.
In fact, as a consequences of the rotational and translation symmetries, the motion starting in a suitable region of the phase space are always contained between two cylinders. In this case, time dependence introduces new resonances in the system that must be taken into account.\\
Notably, in the present case of two wires, the rotation symmetry is broken. However, we are still able to find a confinement region. The key geometrical observation is that the plane containing the wires is invariant. Our confinement region is then a planar strip parallel to the wires and contained in this invariant plane. The location of the strip depends on the currents and can be either between the wires or outside this region. In every case, motions starting inside the strip are bounded in the component orthogonal to the wires.
\par 
To provide closed-form expressions, we will focus on the invariant plane. The dynamics are now described through a time dependent system with two degree of freedom. The translational symmetry allows a further reduction to a time dependent Hamiltonian system with one degree of freedom depending on a parameter given by the linear momentum. The reduced equations describe the motions in the orthogonal direction to the wires. Stability is then understood as Lyapunov stability of the orthogonal component. We stress the fact that in our concept of stability we allow also variations of the linear momentum. Hence we get a stronger concept of stability than the classical iso-energetic stability for Hamiltonian systems.
\par 
To advance our analysis towards our final result we start with fixed values of the linear momentum. This allows us to get periodic solutions for non stationary currents bifurcating from the equilibrium present in the stationary case. At this stage we exclude some resonances and it will be important to assume that the time-dependent currents are small perturbations of constant currents. We then prove that the just obtained periodic solutions is of twist type through the third approximation method, introduced by Ortega in \cite{Ortega} and sharpened by Zhang \cite{zhang}. More precisely we are going to use a version in \cite{torres_zhang}. We recall here that periodic solutions of twist type are Lyapunov stable and accumulated by quasiperiodic and subharmonic solutions. This technique has also been used for the case of one wire in \cite{garzon} and in \cite{lei_zhang} for the case of a time dependent charge density. Finally, adapting an argument in \cite{SM} we also prove that stability of twist type for fixed values of the momentum implies stability allowing variations of the linear momentum. 
\par 
The remainder of this paper is structured as follows. In section \ref{sec:statement} we give a precise formulation of our problem and state our main result together with our concept of stability. In section \ref{sec:em} we study the electromagnetic field generated by the two wires, while the equations of motion are deduced in section \ref{sec:equations}. The study of the existence of the confinement zone is developed in section \ref{sec:stab}. Finally we draw some conclusion and propose a discussion o possible applications and future work.         
\section{Statement of the problem and main result}\label{sec:statement}

In a reference frame $Oxyz$ (see Fig.\ref{fig:twowires}), consider a system of two electrically neutral vertical parallel wires of infinite length with constant distance $d$ between them and carrying non stationary currents $\mathcal{I}^{(i)}_\eta(t)$, where for $i=1,2$,
\begin{equation}
	\mathcal{I}^{(i)}_{\eta}(t)=I^{(i)}_{0}+\eta I^{(i)}(t).  
\end{equation}
Here, $I^{(i)}_{0}$ are constants, $I^{(i)}:\mathbb{R}\rightarrow\mathbb{R}$ are $T$-periodic and $C^4$ functions such that
\begin{equation}
	\int_0^TI^{(i)}(t) dt =0
\end{equation}
and $\eta\geq 0$ is a parameter.
Without loss of generality we suppose the first wire located along the $Oz$ axis and the second passing through the point $(0,d,0)$ (see Fig.\ref{fig:twowires}).
\begin{figure}[h!]
	\begin{center}
		\includegraphics[scale=0.4]{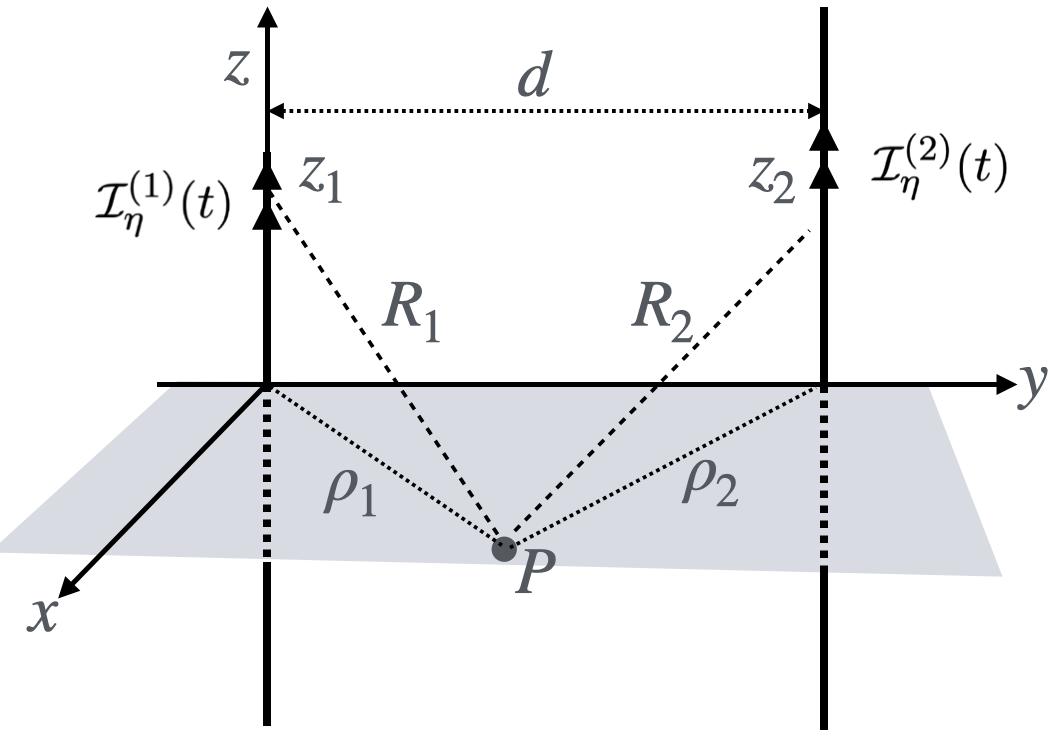}
	\end{center}
	\caption{Schematic of geometric arrangement of two equidistant infinite electrically neutral currents producing a field at point $P$.}
	\label{fig:twowires}
\end{figure}
The system generates an electromagnetic field denoted by $\mathbf{E}_\eta$ and $\mathbf{B}_\eta$ and a moving particle of charge $q$ and mass $m$ suffers the Lorentz force given by:
$$\mathbf{F}_\eta(\mathbf{x},\dx,t)=q(\mathbf{E}_\eta(\mathbf{x},t)+\dx\times\mathbf{B}_\eta(\mathbf{x},t)).$$  
We are interested in the motion of this particle under the Newton-Lorentz equation
\begin{equation}\label{nleq}
	m\ddot{\mathbf{x}}=\mathbf{F}_\eta(\mathbf{x},\dx,t).
\end{equation}
To state our result, we note that by the translation symmetry, the linear momentum $p_z=m\dot{z}+q A_\eta(x,y,t)$ is preserved. See section \ref{sec:em} for the definition of the vector potential $\mathbf{A}_\eta$.  In section \ref{sec:equations} we will see that the plane $Oyz$ containing the wires is invariant and we restrict to motions in this plane. Moreover, in section \ref{sec:stab_aut} we will see that for $\eta=0$ it is possible to find a uniform motion of the particle. Actually, there exists a vertical line $y=\bar{y}$ such that each motion starting on this line with momentum $m\dx = p_z \mathbf{\hat{z}}$ satisfy for every $t\in\mathbb{R}$
\[
y(t)\equiv \bar{y}\qquad\text{and}\qquad z(t)=\frac{p_z}{m}t+z(0).
\] 
We are interested in the stability of motions close to this line for the perturbed case $\eta>0$. More precisely we are interested in the existence of a stable strip of motion, defined as follows. 
\begin{defin}\label{def_stab} 
	For a value of $p_z$, and given $\epsilon>0$  we say that the system of two wires has a $(\epsilon,p_z)$-stable strip around $\bar{y}$, with $\bar{y}\neq 0,d$, if there exists $\delta>0$ such that each solution of \eqref{nleq} with linear momentum $\tilde{p}_z$ and initial condition satisfying 
	\[
	|\dot{y}(0)|+|y(0)-\bar{y}|+|\tilde{p}_z-p_z|<\delta \quad\text{and}\quad (x(0),\dot{x}(0))=(0,0)  
	\] 
	is such that, for every $t\in\mathbb{R}$
	\[
	x(t)\equiv 0, \qquad |\dot{y}(t)|+|y(t)-\bar{y}|<\epsilon \quad\text{and}\quad  \left|m\dot{z}(t)-p_z\right| < C\epsilon 
	\]
	for a constant $C$ depending on $\bar{y}$.
\end{defin} 
\begin{remark}\begin{itemize}
		
		\item After reducing to the invariant plane, the condition for stability means Lyapunov stability in the $y$ component with variation of the initial condition in the whole phase space. In this sense our definition is stronger than the concept of isoenergetic stability where the value of the integral $p_z$ is kept fixed. However, we cannot guarantee stability in the phase space also for the coordinate $z$. Actually, if $p_z\neq 0$ we get uniform motion with positive velocity. 
		
		\item Note that by translational symmetry the role of $z(0)$ is irrelevant. 
		\item We say that this scenario is stable since each motion in the plane $Oyz$ starting close to the line $y=\bar{y}$ with small velocity $\dot{y}$ takes place between two lines of the form $y=C_1$ and $y=C_2$ for two constants $C_1,C_2$ depending on the initial condition, while the motion in the $z$ component is close to a uniform motion with velocity $p_z/m$.
	\end{itemize}
\end{remark}		
To state our main result we introduce the ratio
\[
k=\frac{I^{(2)}_{0}}{I^{(1)}_{0}}
\]
that, since we will suppose $I^{(1)}_{0},I^{(2)}_{0}\neq 0$ is a non-null real number. Using it we denote
\[
I^{(1)}_{0}=I_0, \qquad I^{(2)}_{0}=kI_0.
\]

\begin{theorem}\label{main}
	Suppose that
	\begin{align}\label{stab}
		0<-\frac{q\mu_0}{2\pi m^2}p_zI_0\frac{(1+k)^3}{kd^2}<\left(\frac{\pi}{2T}\right)^2, \quad  
		\frac{9q\mu_0}{4\pi}I_0(1+k)\neq \frac{p_z}{k}(-k^2+11k-1)
	\end{align}
	and
	\begin{equation}\label{nonres}
		-\frac{q\mu_0}{2\pi m^2}p_zI_0\frac{(1+k)^3}{kd^2}\neq \left(\frac{2\pi}{nT}\right)^2 \qquad\text{for all } n\in\mathbb{N}.
	\end{equation}
	Then, for every $\epsilon>0$ there exists $\eta_2>0$ such that the system of two wires for $\eta<\eta_2$ has a $(\epsilon,p_z)$-stable strip around
	\[
	\bar{y}=\frac{d}{1+k}.
	\]    
\end{theorem}

Note that, depending on the value of the parameter $k$, the strip can be located either between the two wires or outside this region. Moreover, as a byproduct of the proof we will also get the following
\begin{remark}
	\begin{itemize}
		\item If $\eta=0$ and $p_zI_0\frac{1+k}{k}<0$ then the $y$-components of all the solutions with initial condition in the stable strip are periodic.
		\item Under the conditions of Theorem \ref{main} and for $\eta<\eta_2$ there exist, in the stable strip, infinitely many solutions whose $y$-component is subharmonic and an uncountable collection of solutions whose $y$-component is quasiperiodic. To avoid some technicality, we refer to \cite{SM,OrtegaLisboa} for the formal definition of these solutions.   
	\end{itemize}	
\end{remark}

\section{The electromagnetic field}\label{sec:em}
In this section, following \cite{garzon}, we build the electromagnetic field generated by our system. The electric field $\mathbf{E}_\eta$ and the magnetic field $\mathbf{B}_\eta$ must be solutions of the Maxwell equations with charge density $\rho_\eta=0$ and $\mathbf{J}_\eta$ given by:
\begin{equation}
	\mathbf{J}_\eta(\mathbf{x},t) = \mathcal{I}^{(1)}_{\eta}(t)\delta(x)\delta(y)\mathbf{\hat{z}} +\mathcal{I}^{(2)}_{\eta}(t)\delta(x)\delta(y-d)\mathbf{\hat{z}}.
\end{equation}

This is accomplished by introducing the respective electric and magnetic potentials:
\begin{eqnarray*}
	\mathbf{E}_\eta=-\nabla\phi_\eta-\frac{\partial \mathbf{A}_\eta}{\partial t}, \qquad
	\mathbf{B}_\eta=\nabla\times\mathbf{A}_\eta,
\end{eqnarray*}
and the Lorenz Gauge. We can choose the scalar potential $\phi_\eta=0$. The vector potential is given by the Jefimenko formula:
\begin{equation}\label{Jefi}
	\mathbf{A}_\eta(\mathbf{x},t)=\frac{\mu_0}{2\pi}\mathbf{\hat{z}}\int_{0}^{\infty}\frac{\mathcal{I}^{(1)}_\eta(t_r^{(1)})}{R_{1}}dz+\frac{\mu_0}{2\pi}\mathbf{\hat{z}}\int_{0}^{\infty}\frac{\mathcal{I}^{(2)}_\eta(t_r^{(2)})}{R_{2}}dz
\end{equation}
where $\mu_0$ is the magnetic permittivity in the vacuum and $t^{(i)}_{r}=t-R_{i}/c$ is the retarded time. Here we denote $R_i=\sqrt{z^2+\rho_i^2}$ with  $\rho^2_{1}=x^2+y^2$, and $\rho^2_{2}=x^2+(y-d)^2$, see Fig.\ref{fig:twowires}.\\
Note that the case $\mathbf{A}_0$ produces divergent integrals. However it is possible to see that Jefimenko's renders a solution of the corresponding wave equation in the distributional sense (see \cite{garzon}).\\
The total vector potential is independent on $z$ and given by $\mathbf{A}_\eta(\mathbf{x},t)=A_\eta(x,y,t)\mathbf{\hat{z}}$, where 
\[
A_\eta(x,y,t)=A_0(x,y)+\eta A(x,y,t).
\] 
\[
A_0(x,y) =-\frac{\mu_0}{2\pi}\sum_{i=1}^2I_0^{(i)}\ln(\rho_i) , \qquad  A(x,y,t)=-\frac{\mu_0}{2\pi}\sum_{i=1}^2 a_i(x,y,t)
\]
where 
\[
a_i(x,y,t)=\int_0^\infty  \frac{I^{(i)}\left(t_r^{(i)}\right)}{R_i} dz.
\]
Here $a_1$ is $C^4(\{\mathbb{R}^2\setminus (0,0)\}\times\mathbb{R})$, $a_2$ is $C^4(\{\mathbb{R}^2\setminus (0,d)\}\times\mathbb{R})$ and both are $T$-periodic in $t$.

\noindent The electric and magnetic fields are then given by 

\begin{eqnarray}
	\nonumber\mathbf{E}_\eta(\mathbf{x},t)= -\frac{\partial \mathbf{A}_\eta(\mathbf{x},t)}{\partial t}=
	-\frac{\partial A_\eta(x,y,t)}{\partial t}\mathbf{\hat{z}} =-\eta\frac{\partial A(x,y,t)}{\partial t}\mathbf{\hat{z}}= \eta E(x,y,t)\mathbf{\hat{z}},\\
	\mathbf{B}_\eta(\mathbf{x},t)=\nabla\times\mathbf{A}_\eta(\mathbf{x},t)=
	\nabla\times A_\eta(x,y,t)\mathbf{\hat{z}} =-\frac{\partial A_\eta}{\partial x}(x,y,t)\mathbf{\hat{y}} +\frac{\partial A_\eta}{\partial y}(x,y,t)\mathbf{\hat{x}}. 
	\label{eq:EB}  
\end{eqnarray}

When $I^{(i)}(t)$ are linear combinations of sines an cosines, these fileds have an explicit expression in terms of Bessel and Neumann functions (see \cite{AG}).

\section{Equations of motion}\label{sec:equations}

The Newton-Lorentz equation \eqref{nleq} is given by the Euler-Lagrange equations for the Lagrangian
\begin{equation}
	L(\mathbf{x},\mathbf{\dot{x}},t)=\frac{1}{2}m|\mathbf{\dot x}|^2-q\phi_\eta+q\mathbf{A}_\eta\cdot\dx
\end{equation}
which results in
\begin{equation}
	L(\mathbf{x},\mathbf{\dot{x}},t)=\frac{1}{2}m|\mathbf{\dot x}|^2+q \dot{z}A_\eta(x,y,t). 
\end{equation} 
Note that the variables $z$ is cyclic, so that the conjugated momenta $p_z$ is a first integral. This will be evident in the following Hamiltonian formulation.
To this aim, we introduce the generalized momenta by
\[
\mathbf{p}=\frac{\partial L}{\partial \mathbf{\dot{x}}} = m\mathbf{\dot{x}}+q\mathbf{A}_\eta.
\]
The Hamiltonian is given by
\begin{align*}
	H(\mathbf{x},\mathbf{p},t)&=\frac{1}{2m}\left|\mathbf{p}-q\mathbf{A}_\eta(\x,t)\right|^2+q\phi_\eta(\x,t)\\
	&=\frac{1}{2m}\left\{p_x^2+p_y^2+[p_z-qA_\eta(x,y,t)]^2\right\}
\end{align*}
and the Hamilton equations are
\begin{align*}
	\dot{x}&=\frac{1}{m}p_x, & \dot{p}_x &=\frac{q}{m}[p_z-qA_\eta(x,y,t)]\frac{\partial A_\eta}{\partial x}(x,y,t),\\
	\dot{y}&=\frac{1}{m}p_y, & \dot{p}_y&= \frac{q}{m}[p_z-qA_\eta(x,y,t)]\frac{\partial A_\eta}{\partial y}(x,y,t),\\ 	
	\dot{z}&=\frac{1}{m}[p_z-qA_\eta(x,y,t)], & \dot{p}_z&= 0, 	 
\end{align*}
from which we get that $p_z$ is constant and
\begin{align}
	\ddot{x} &=\frac{q}{m^2}[p_z-qA_\eta(x,y,t)]\frac{\partial A_\eta}{\partial x}(x,y,t)\label{xeq}, \\
	\ddot{y} &=\frac{q}{m^2}[p_z-qA_\eta(x,y,t)]\frac{\partial A_\eta}{\partial y}(x,y,t) \label{yeq}, \\
	\dot{z}&=\frac{1}{m}[p_z-qA_\eta(x,y,t)]\label{zeq}. 	 
\end{align}
Hence the dynamics is reduced to equations \eqref{xeq}-\eqref{yeq} depending on the parameter $p_z$. Note that this is a time-dependent Hamiltonian system with two degrees of freedom, making it complex to analyze. We can reduce the degrees of freedom noting that the plane $Oyz$ is invariant and considering the motions that take place in this plane. This is the aim of the following Lemma  

\begin{lemma}\label{lemmaz0}
	Each solution of $\eqref{xeq},\eqref{yeq},\eqref{zeq}$ satisfying $x(t_0)=\dot{x}(t_0)=0$ for some $t_0\in\mathbb{R}$ is such that 
	\[
	x(t)\equiv 0 \qquad\text{for all } t\in\mathbb{R}. 
	\] 
\end{lemma}
\begin{proof}
	We claim that for all $y\neq{0,d}$ and $t\in\mathbb{R}$,
	\begin{equation}\label{condinv}
		\frac{\partial A_\eta}{\partial x}(0,y,t) =0.
	\end{equation}
	From a direct computation we get that 
	\[
	\frac{\partial A_{0}}{\partial x}(0,y)=0.
	\]
	Concerning the term $A(x,y,t)$, we differentiate $a_i(x,y,t)$, $i=1,2$, under the integral and get
	\begin{align*}
		\frac{\partial a_i}{\partial x}(x,y,t)&=\int_0^\infty  \frac{\partial}{\partial x}\frac{I^{(i)}\left(t_r^{(i)}\right)}{R_i} dz\\ & =
		-\int_0^\infty x\left[ \frac{\dot{I}^{(i)}dt\left(t_r^{(i)}\right)}{cR_i^2}+ \frac{I^{(i)}\left(t_r^{(i)}\right)}{2R_i^{3/2}}\right] dz.
	\end{align*}
	Evaluating in $x=0$ we get the claim proved. The thesis follows from the uniqueness of the solution and the fact that $x(t)\equiv 0$ is a solution of \eqref{xeq}.  
\end{proof}

\begin{remark}
	Note that condition \eqref{condinv} implies that for $x=0$ the magnetic filed $\mathbf{B}_\eta(\mathbf{x},t)$ give by \eqref{eq:EB} is orthogonal to the plane $Oyz$. Hence, the Lorentz force $\mathbf{F}(\mathbf{x},\dx,t)$ gives no contribution in the direction $\mathbf{\hat{x}}$ if also $\dot{x}=0$.   
\end{remark}

From Lemma \ref{lemmaz0} we can reduce equations \eqref{yeq} and \eqref{zeq} fixing $x=0$, getting
\begin{align}
	\ddot{y} &=\frac{q}{m^2}[p_z-qA_\eta(y,t)]\frac{\partial A_\eta}{\partial y}(y,t)\label{yeq1} \\
	\dot{z}&=\frac{1}{m}[p_z-qA_\eta(y,t)]\label{zeq1} 
\end{align}
where we have denoted $A_\eta(y,t)=A_\eta(0,y,t)$.

We can solve the reduced system once we have the solution of the equation \eqref{yeq1} that can be rewritten as

\begin{equation}\label{yeq12}
	\ddot y +V'(y)+\eta F(y,t;\eta)=0
\end{equation}
with:
\begin{align*}
	V(y)&=\frac{1}{2m^2}[p_z-qA_0(y)]^2 \\ 
	F(y,t;\eta) &= \frac{q}{m^2}\left\{qA_0'(y)A(y,t)-[p_y-qA_0(y)]A'(y,t)+\eta q A(y,t)A'(y,t)\right\}
\end{align*}

\section{Existence of the stable strip of motion}\label{sec:stab}

\subsection{The autonomous case}\label{sec:stab_aut}

Suppose that $\eta=0$ so that the problem is autonomous and the motion of the $y$-component is given by
\begin{equation}\label{yeq120}
	\ddot y +V'(y)=0
\end{equation}
that is an oscillator with potential $V(y)$. We have that, 
\[
V'(y)=-\frac{q}{m^2}[p_z-qA_0(y)]A'_0(y) 
\]
and we can have stationary solutions $y(t)=\bar{y}$ in two cases. The first case is when $m\bar{\dot{z}}=p_z-qA_0(\bar{y})=0$ and corresponds to the case of null velocity (and hence null forces). This is a degenerate case since all the points of the $Oyz$ plane are equilibria. We will be interested in the case when $A'_0(\bar{y})=0$, that, if $k\neq -1$ is the weighted mean of the distance between the wires
\[
\bar{y}=d\frac{I_0^{(1)}}{I_0^{(1)}+I_0^{(2)}}=\frac{d}{1+k}
\]
where the magnetic field $\mathbf{B}_0$ vanishes.  Since $A_0(y)$ is a potential, it is defined up to an additive constant so that we can suppose that $A_0(\bar{y})=0$. Hence, to exclude the previous degenerate case we suppose $m\bar{\dot{z}}=p_z\neq 0$. Since
\[
V''(\bar{y})=-\frac{q}{m^2}(p_z-qA_0(\bar{y}))A_0''(\bar{y})=-\frac{q}{m^2}p_zA_0''(\bar{y}),
\] 
if $p_zA_0''(\bar{y})>0$ then the stationary solution is a maximum of the effective potential and it is unstable. On the other hand, if
\[
qp_zA_0''(\bar{y})=
q\frac{\mu_0}{2\pi}p_zI_0\frac{(1+k)^3}{kd^2}<0
\] 
then the stationary solution is a minimum of the effective potential and it is stable and surrounded by a continuum of periodic solutions $y_\tau(t)$ with periods 
\[
\tau\to \frac{2\pi}{\sqrt{V''(\bar{y})}}.
\] 
\begin{remark}
	Note that the $y$-component of the Lorentz force is $A'_0(y)\dot{z}$. The conditions $A'_0(\bar{y})=0$ and $q\bar{\dot{z}}A_0''(\bar{y})<0$ for stability imply that $y$-component of the Lorentz force close to the line $y=\bar{y}$ is directed towards the same line.
\end{remark}

\begin{remark}
	Note that if $k=-1$ and $\dot{z}\neq 0$ then $V'(y)$ has constant sign. This implies that equation \eqref{yeq120} has no constant solutions and all the solutions are unbounded.  
\end{remark}
We have that
\begin{theorem}
	Suppose that
	\begin{equation}\label{condpz}
		qp_zI_0\frac{1+k}{k}<0.
	\end{equation}
	Then the system of two wires for $\eta=0$ has a $p_z$-stable strip around
	\[
	\bar{y}=\frac{d}{1+k}.
	\]
	
\end{theorem}
\begin{proof}
	We have already noticed that if $(x(0),\dot{x}(0))=(0,0)$ then $x(t)\equiv 0$. Moreover, the phase space analysis just performed for \eqref{yeq120} and fixed $p_z$ satisfying \eqref{condpz} implies that $(\bar{y},0)$ is a stable center in the phase-plane $(y,p_y)$. To show stability in the sense of Definition \ref{def_stab} we need to consider variations also of the parameter $p_z$. The Hamiltonian of \eqref{yeq120} is
	\begin{equation}\label{hy120}
		H_{p_z}(y,p_y)=\frac{1}{2m}[p_y+(p_z-qA_0(y))^2]
	\end{equation}
	where the constant momentum $p_z$ is seen as a parameter. It is constant along the solutions and has a minimum independent on $p_z$ at $(\bar{y},0)$. Given $\epsilon>0$ and fixing $p_z$, we can find an invariant region containing $(\bar{y},0)$ and of radius smaller than $\epsilon$ of the form $H_{p_z}(y,p_y)\leq \rho(p_z,\epsilon)$. By continuity w.r.t. $p_z$, there exists $\delta<\epsilon$ such that if $|\tilde{p}_z-p_z|<\delta$ we can do the same for the Hamiltonian $H_{\tilde{p}_z}(p_y,y)$. This implies that the region
	\[
	\mathcal{U}=\{(y,p_y,z,\tilde{p}_z)\: :\: H_{\tilde{p}_z}(p_y,y)\leq \rho(\tilde{p}_z,\epsilon),|\tilde{p}_z-p_z|<\delta \}\subset \mathbb{R}^4    
	\]
	is invariant under the flow of the Hamiltonian \eqref{hy120} seen in $\mathbb{R}^4$. This implies the stability of the $y$ component as required in Definition \ref{def_stab}.
	Let us now consider a solution $(y(t),z(t))$ with initial condition in $\mathcal{U}$. If $y(t)=\bar{y}$, substituting in \eqref{zeq1} we get 
	\[
	z(t)=\frac{\tilde{p}_z}{m}t+z(0).
	\] 
	In the other cases, $y(t)$ is periodic with period $\tau$. We denote it $y_\tau(t)$. Substituting in \eqref{zeq1} we get 
	\begin{equation}\label{stab_z0}
		m\dot{z}(t)=\tilde{p}_z-qA_0(y_\tau(t))
	\end{equation}
	By the stability in the $y$ component, we write $y_\tau(t)=\bar{y}+\xi_\tau(t)$ where $\xi_\tau(t)$ is $\tau$-periodic and
	\[
	|\dot{\xi}_\tau(t)|+|\xi_\tau(t)|<\epsilon\qquad \text{for every }t\in\mathbb{R}
	\]
	and expand $A(y_\tau(t))$ as
	\begin{equation}\label{expand}
		A_0(y_\tau(t))=A_0(\bar{y})+A_0'(\bar{y})\xi_\tau(t)+\frac{1}{2}A_0''(\bar{y})\xi_\tau^2(t)+R(t)=\frac{1}{2}A_0''(\bar{y})\xi_\tau^2(t)+R(t)
	\end{equation}
	where we have used that in our case $A_0(\bar{y})=A_0'(\bar{y})=0$ and the remainder satisfies
	\[
	|R(t)| \leq C|\xi_\tau(t)|^3 \quad\text{for some } C\in\mathbb{R}.
	\] 
	Inserting in \eqref{stab_z0} we get
	\[
	|m\dot{z}(t)-p_z|=|\tilde{p}_z-p_z|+|qA_0(y_\tau(t))|<\epsilon+C\epsilon^2
	\] 
	From this we get the required estimate on $\dot{z}(t)$ and the thesis.
	
\end{proof}
The following picture can be deduced:
\begin{remark} 	
	If $\eta=0$ and $(x(0),\dot{x}(0))=(0,0)$ the motion takes place in the plane $Oyz$. Fix $p_z$ satisfying \eqref{condpz} and initial conditions $(y(0),\dot{y}(0),\tilde{p}_z)$ in a $3$ dimensional neighborhood of $(\bar{y},0,p_z)$ and $z(0)\in\mathbb{R}$.  
	If $(y(0),\dot{y}(0))=(\bar{y},0)$ then the motion takes place on the straight line  $y=\bar{y}$ with constant velocity $\tilde{p}_z$. If $(y(0),\dot{y}(0))\neq (\bar{y},0)$ then the $y$-component oscillates with period $\tau$ (depending continuously on the initial conditions) around $y=\bar{y}$ and the linear increase of $z$ suffers a $\tau$-periodically perturbation.       
\end{remark}

We are interested in the behavior of this scenario under perturbations as we elaborate below. 
\subsection{The perturbed case}\label{sec:stab_pert}

Note that, if $\eta\neq 0$, then equation \eqref{yeq12} is $T$-periodic in $t$. Hence, the only periodic solutions can have period of the form $nT$ with $n\in\mathbb{N}$. 

We can prove that there exists a periodic solution bifurcating from the constant solution
$\bar{y}$ found before in the unperturbed case.
To prove it, we need to assume the period $nT$ to be away from resonances. These correspond to the period of the solutions of the linearization of \eqref{yeq12} for $\eta=0$ at $\bar{y}$
\[
\ddot y +V''(\bar{y})y=0, 
\]
namely
\[
\frac{2\pi}{\sqrt{V''(\bar{y})}}
\]
We recall that we have supposed
\[
V''(\bar{y})=-\frac{q\mu_0}{2\pi m^2}p_zI_0\frac{(1+k)^3}{kd^2}>0.
\]
From the general theory of ordinary differential equations \cite{Cod} we have 
\begin{theorem}\label{teoper}
	Suppose that
	\[
	\frac{2\pi}{\sqrt{V''(\bar{y})}}\neq nT \qquad\text{for all } n\in\mathbb{N}.
	\]
	Then there exists $\eta_0>0$ such that for all $|\eta|<\eta_0$ there exists a $T$-periodic solution $y_\eta(t)$ of \eqref{yeq12}. The family of solutions $y_\eta(t)$ satisfies
	\[
	\lim_{\eta\to 0} y_\eta(t) =\bar{y} \qquad\text{uniformly in }t.
	\] 	
\end{theorem}

We are now interested in the Lyapunov stability of the periodic solutions $y_\eta(t)$ we just found. More precisely we are going to find conditions under which these solutions are of twist type. The definition of periodic solution of twist type is given in \cite{SM} and we do not recall it to avoid technicalities. We just highlight that a periodic solution of twist type is Lyapunov stable and accumulated by subharmonic and quasiperiodic solutions. \\
To prove that the periodic solutions $y_\eta(t)$ are of twist type we follow the approach in \cite{Ortega} and subsequent generalizations \cite{zhang, torres_zhang}. 
This depends on the third approximation of \eqref{yeq12} around $y_\eta(t)$. To compute it, let us perform the change of variable $Y=y-y_\eta(t)$ and expand up to third order around $Y=0$, getting
\[
\ddot{Y}+\alpha_\eta(t)Y + \beta_\eta(t)Y^2 +\gamma_\eta(t)Y^3+\dots =0
\]  
where
\begin{align*}
	\alpha_\eta(t)&= V''(y_\eta(t)) + \eta\frac{\partial}{\partial y} F(y_\eta(t),t;\eta)      \\
	2\beta_\eta(t)&= V'''(y_\eta(t)) +  \eta\frac{\partial^2}{\partial y^2} F(y_\eta(t),t;\eta)    \\
	6\gamma_\eta(t)&= V^{(iv)}(y_\eta(t))+  \eta\frac{\partial^3}{\partial y^3} F(y_\eta(t),t;\eta).    
\end{align*}    
Since $y_\eta(t)\rightarrow \bar{y}$ uniformly in $t$ as $\eta\to 0$, we can write  
\begin{align*}
	\alpha_\eta(t)&=\alpha_0 + \xi_\alpha(t,\eta)     \\
	\beta_\eta(t)&=\beta_0  +\xi_\beta(t,\eta) \\
	\gamma_\eta(t)&=\gamma_0 +\xi_\gamma(t,\eta)
\end{align*}
where $\xi_\alpha,\xi_\beta,\xi_\gamma$ are $T$-periodic in $t$ and tend to $0$ uniformly as $\eta\to 0$. 
From a direct computation we have that
\begin{align*}
	\alpha_0 =V''(\bar{y}) & =-\frac{q}{m^2}p_zA_0''(\bar{y})= -\frac{q\mu_0}{2\pi m^2}p_zI_0\frac{(1+k)^3}{kd^2}  \\
	\beta_0 =\frac{1}{2}V'''(\bar{y})&=  -\frac{q}{2m^2}p_zA_0'''(\bar{y})= \frac{q\mu_0}{2\pi m^2}p_zI_0\frac{(1+k)^4(k-1)}{k^2d^3}\\
	\gamma_0 =\frac{1}{6}V^{(iv)}(\bar{y})&= \frac{q^2}{2m^2}A_0''(\bar{y})-\frac{q}{6m^2}p_zA_0^{(iv)}(\bar{y}) \\
	&=\frac{q\mu_0}{2\pi m^2}I_0^2\frac{(1+k)^4}{d^4k^2}\left[\frac{q\mu_0}{4\pi}(1+k)^2-p_z\frac{1+k^3}{kI_0}\right]
\end{align*}

%

From \cite{torres_zhang} we have that the trivial solution $Y=0$ (and hence the $T$-periodic solution $y_\eta(t)$) is of twist type if
\begin{equation}\label{rafateo}
\begin{split}
		&0<(\alpha_\eta)_*\leq (\alpha_\eta)^*<\left(\frac{\pi}{2T}\right)^2, \\
		&(\gamma_\eta)_*>0, \\
		&10(\beta_\eta)_*^2(\alpha_\eta)_*^{3/2}>9(\gamma_\eta)^*[(\alpha_\eta)^*]^{5/2}\quad \text{or} \quad 10(\beta_\eta^*)^2(\alpha_\eta^*)^{3/2}<9(\gamma_\eta)_*(\alpha_\eta)_*^{5/2} 
\end{split}
\end{equation}
where $f^*$ and $f_*$ respectively represent the supremum and infimum of $f\in\mathcal{C}\left([0,T];\mathbb{R}\right)$.
Summing up, we are ready for the

\begin{proof}[Proof of Theorem \ref{main}]
	Condition \eqref{nonres} implies that we can apply Theorem \ref{teoper} and get the periodic solutions $y_\eta(t)$ for $\eta<\eta_0$ such that
	\begin{equation}\label{unifep}
		\lim_{\eta\to 0} y_\eta(t) =\bar{y} \qquad\text{uniformly in }t.
	\end{equation}
	The first assumption in \eqref{stab} implies that both the first and the second condition in \eqref{rafateo} hold for $\alpha_0,\gamma_0$. The second hypothesis in \eqref{stab} implies that one of the two possibilities in the last condition in \eqref{rafateo} holds for $\alpha_0,\beta_0,\gamma_0$. Note that it is enough to check that $10\beta_0^2\neq 9\gamma_0\alpha_0$.  Since the remainders $\xi_\alpha,\xi_\beta,\xi_\gamma$ tend to $0$ uniformly as $\eta\to 0$ we get that there exists $\eta_1\leq\eta_0$ such that the periodic solutions $y_\eta(t)$ are of twist type $\eta<\eta_1$. 
	
	We prove that this implies the existence of the $p_z$-stable strip. We will proceed adapting an argument in \cite{SM} (see also \cite{garzon}). Fix $p_z$ satisfying \eqref{stab}. Consider the map $P_\eta(y_\eta(0),\dot{y}_\eta(0),p_z)=(y_\eta(T),\dot{y}_\eta(T),p_z)$ that maps one of these solutions with initial condition $(y_\eta(0),\dot{y}_\eta(0))$ and first integral $p_z$ to the corresponding values at time $T$. We claim that $q_\eta=(y_\eta(0),\dot{y}_\eta(0),p_z)$ is a stable fixed point for the map $P$. Actually, let us fix a neighborhood $\mathcal{U}$ of $q_\eta$. Since the point $(y_\eta(0),\dot{y}_\eta(0))$ is the initial condition of a periodic solution of \eqref{yeq12} of twist type for fixed values of $p_z$, we can find in $\mathcal{U}$ an invariant planar (i.e. with fixed $p_z$) region bounded by a closed curve surrounding $(y_\eta(0),\dot{y}_\eta(0))$ of the form
	\[
	\mathcal{U}_1=\{(y-y_\eta(0))^2+(\dot{y}-\dot{y}_\eta(0))^2\leq R(\Theta), \: p_z=p_{z_0}  \}\subset\mathcal{U},
	\]
	where $\Theta=\Theta(y_\eta(0),\dot{y}_\eta(0))$ represents the angle centered in $(y_\eta(0),\dot{y}_\eta(0))$ and $R(\cdot)$ represent the parametrization of such curve.
	Equation \eqref{yeq12} depends continuously on the parameters $p_z$, and condition \eqref{stab} holds in an open set. Hence by continuous dependence, we can find a family of curves
	\[
	(y-y_\eta(0))^2+(\dot{y}-\dot{y}_\eta(0))^2\leq R(\Theta,p_z)
	\]
	depending continuously on $p_z$ and with the properties above. Therefore, for sufficiently small $\delta$, the region
	\[
	(y-y_\eta(0))^2+(\dot{y}-\dot{y}_\eta(0))^2\leq R(\Theta,p_z), \: \: |\tilde{p}_z-p_{z}|<\delta \}
	\]
	is invariant and contained in $\mathcal{U}$. This implies that the point $q_\eta$ is stable under the map $P$. 
	Now fix $\epsilon>0$. By \eqref{unifep} we can find $\eta_{2}$ such that for $\eta<\eta_{2}$ 
	\[
	|\dot{y}_\eta(t)|+|y_\eta(t)-\bar{y}|<\epsilon/2.
	\]   
	By the stability of $q_\eta$, there exists $\delta$ such that if 	
	\begin{equation}\label{deltastab}
		|\dot{y}(0)-\dot{y}_\eta(0)|+|y(0)-y_\eta(0)|+|\tilde{p}_z-p_z|<\delta  
	\end{equation}
	then the solution of \eqref{xeq},\eqref{yeq},\eqref{zeq} with initial conditions
	$(0,0,y(0),\dot{y}(0),z(0),\tilde{p}_z)$, $z(0)\in\mathbb{R}$ satisfies
	\[
	x(t)\equiv 0, \qquad |\dot{y}(t)-\dot{y}_\eta(t)|+|y(t)-y_\eta(t)|<\epsilon/2 \quad \text{for every }t\in\mathbb{R}.  
	\]  
	This last two inequalities imply that, for $\eta<\eta_2$, the solutions with initial condition 
	$(0,0,y(0),\dot{y}(0),z(0),\tilde{p}_z)$, $z(0)\in\mathbb{R}$ satisfying \eqref{deltastab} satisfies 
	\[
	x(t)\equiv 0, \qquad |\dot{y}(t)|+|y(t)-\bar{y}|<\epsilon \quad \text{for every }t\in\mathbb{R}  
	\]
	that is the stability of the $y$ component in Definition \ref{def_stab}. To prove the condition on the $z$-component, we substitute this solutions in \eqref{zeq1} to get 
	\begin{equation}\label{stab_ze}
		m\dot{z}(t)=\tilde{p}_z-qA_0(y(t))-q\eta A(y(t))
	\end{equation}
	By the stability in the $y$ component, we write $y(t)=\bar{y}+\xi(t)$ with 
	\[
	|\dot{\xi}(t)|+|\xi(t)|<\epsilon\qquad \text{for every }t\in\mathbb{R}
	\]
	and expand $A_0(y(t))$ as in \eqref{expand}. Inserting in \eqref{stab_ze} we get
	\[
	|m\dot{z}(t)-p_z|\leq|\tilde{p}_z-p_z|+|qA_0(y(t))|+\eta \max_{|y-\bar{y}|<\epsilon}q|A(y)|< \epsilon+C\epsilon^2+\eta \max_{|y-\bar{y}|<\epsilon}q|A(y)|
	\]
	from which we get the thesis, eventually decreasing $\eta_2$ such that
	\[
	\eta_2  \max_{|y-\bar{y}|<\epsilon}|A(y)|<\epsilon.
	\]
	
\end{proof}
\section{Discussion and Future Work}

In this work, we have analyzed the stability of a charged particle confined between two infinite conducting wires carrying non-uniform currents. By employing the Newton-Lorentz equation and considering the time-dependent nature of the currents, we derived the conditions under which a stable region of motion exists. Our study identified a confinement zone where the charged particle's motion remains stable despite time-dependent perturbations in the current. Notably, we were able to establish the presence of a stable strip in the plane parallel to the wires, providing a deeper understanding of how non-stationary currents influence particle dynamics. 
The stability has been studied via the third approximation method developed by Ortega \cite{Ortega} and sharpened by Zhang \cite{zhang}. More precisely we have proved the existence of a periodic motion of the orthogonal component to the wires that is also of twist type. This implies the stability of the orthogonal component, in a stronger sense than the classical isoenergetic stability.
\par 

These findings contribute to both theoretical and applied electromagnetism, offering new insights into the design of magnetic confinement systems used in plasma physics and electromechanical devices. The results are relevant for improving the control of charged particles in magnetic confinement devices such as tokamaks and stellarators, as well as for applications in high-voltage transmission lines and electromagnetic levitation systems.

\par

Future research can extend the current study in several directions:
\begin{itemize}
	
	\item Relativistic regime: we studied the Newtonian regime, neglecting the relativistic effects on high velocities. However, it could be interesting to investigate these effects using the model defined in \cite{Planck,Po2} and following \cite{GT2023}.
	\item Complex Current Configurations: One potential avenue is to explore the stability of charged particles in more complex arrangements of currents, such as one-dimensional arrays of multiple wires or three-dimensional structures. This would provide insights into how more intricate magnetic field configurations affect particle confinement. Moreover, the eventual loss of symmetry would require a more complex analysis. 
	\item Stochastic Perturbations: Another promising direction is the inclusion of stochastic perturbations in the current or particle motion. Investigating how random fluctuations in the system influence stability could be crucial for applications where noise or uncertainty plays a significant role.
\end{itemize}

\section*{Funding}
S.M. has been supported by the project PID2021-128418NA-I00 awarded by the Spanish Ministry of Science and Innovation.

\end{document}